\documentclass[journal]{IEEEtran}

\usepackage{cite}
\usepackage{epsfig}
\usepackage{epstopdf}
\usepackage{graphicx}
\usepackage{url}
\usepackage{amsfonts}
\usepackage{amsmath,bm,amsthm}
\usepackage{amssymb}
\usepackage{multicol}
\usepackage{times}
\usepackage{psfrag}
\usepackage{subfigure}
\usepackage{stfloats}
\usepackage{array}
\usepackage{booktabs, threeparttable}
\usepackage{setspace}
\usepackage{color}
\usepackage[utf8]{inputenc}
\usepackage{bm}

\newtheorem{theorem}{Theorem}

\allowdisplaybreaks

\newtheorem{proposition}{Proposition}
\usepackage{algorithm}
\usepackage{algorithmic}

\newcommand{\dop}{_{{\rm D},l}}
\newcommand{\dopz}{_{{\rm D},0}}

\newcommand{\p}{\Omega}

\begin{document}
\title{Uplink Sensing in Perceptive Mobile Networks with Asynchronous Transceivers}
\author{{Zhitong Ni,~\IEEEmembership{Student~Member,~IEEE},
		J. Andrew Zhang,~\IEEEmembership{Senior~Member,~IEEE},\\
         Xiaojing Huang,~\IEEEmembership{Senior~Member,~IEEE},  Kai Yang, ~\IEEEmembership{Member,~IEEE}, and Jinhong Yuan,~\IEEEmembership{Fellow,~IEEE}}
\thanks{ Z. Ni and K. Yang are with the School of Information and Electronics, Beijing
Institute of Technology, Beijing 100081, China. Z. Ni is also with the Global Big
Data Technologies Centre, University of Technology Sydney, Sydney, NSW
2007, Australia (Emails: zhitong.ni@student.uts.edu.au, yangkai@ieee.org).}
\thanks{ J. A. Zhang and X. Huang are with the Global Big Data Technologies Centre, University of Technology Sydney, Sydney, NSW
2007, Australia (Emails:  Andrew.Zhang@uts.edu.au, Xiaojing.Huang@uts.edu.au).}	
\thanks{J. Yuan is with the University of New South Wales, Sydney, NSW
2052, Australia (Email: J.Yuan@unsw.edu.au).}
}	

\maketitle

\begin{abstract}
Perceptive mobile network (PMN) is a recently proposed next-generation network that integrates radar sensing into communication. One major challenge for realizing sensing in PMNs is how to deal with spatially-separated asynchronous transceivers. The asynchrony between sensing receiver and transmitter will cause both timing offsets (TOs) and carrier frequency offsets (CFOs) and lead to degraded sensing accuracy in both ranging and velocity measurements. In this paper, we propose an uplink sensing scheme for PMNs with asynchronous transceivers, targeting at resolving the sensing ambiguity and improving the sensing accuracy. We first adopt a cross-antenna cross-correlation (CACC) operation to remove the sensing ambiguity associated with both TOs and CFOs. Without sensing ambiguity, both actual propagation delay and actual Doppler frequency of multiple targets can be obtained using CACC outputs. To exploit the redundancy of the CACC outputs and reduce the complexity, we then propose a novel mirrored-MUSIC algorithm, which halves the number of unknown parameters to be estimated, for obtaining actual values of delays and Doppler frequencies. Finally, we propose a high-resolution angles-of-arrival (AoAs) estimation algorithm, which jointly processes all measurements from spatial, temporal, and frequency domains. The proposed AoAs estimation algorithm can achieve significantly higher estimation accuracy than that of using samples from the spatial domain only. We also derive the theoretical mean-square-error of the proposed algorithms. Numerical results are provided and validate the effectiveness of the proposed scheme.
\end{abstract}


\section{Introduction}
The emerging joint communication and radar sensing (JCAS) techniques, aka, dual-functional radar-communications (DFRC), integrate communication and radar sensing functions into one system by sharing a single transmitted signal and many hardware and signal processing modules \cite{strum1,chiri17, kumari,andrew19,Luo19,liufan20}. The integration not only achieves immediate benefits of reduced size, power consumption, cost, and improved spectrum efficiency but also helps to establish a communication link using sensing information or vise versa \cite{radaraided}. The perceptive mobile network (PMN) \cite{framework17,lushan,lushan5G,lushanSvy} is a recently proposed next-generation mobile network based on the JCAS techniques. The concept of PMN was first introduced in \cite{framework17} and then elaborated in \cite{lushan}. Evolving from the current communication-only mobile network, PMN is expected to serve as a ubiquitous radar-sensing network, whilst providing uncompromising mobile communication services.

Although PMN and its systematic framework were introduced in \cite{framework17}, JCAS technologies have been actively studied in the past decade, particularly the technologies closely related to modern mobile networks. In \cite{lush11}, the orthogonal-frequency-division-multiplexing (OFDM) signal was used for sensing and communication simultaneously. Also using the OFDM signals, the authors in \cite{gudelay} developed a smoothing approach that jointly estimates the delay and Doppler frequency of targets moving at high speed. In the scenario of multiuser systems, the authors in \cite{strum13} proposed an interleaved OFDM signal model to mitigate the multiuser interference (MUI). The authors in \cite{sit} analyzed the MUI tolerance of a multiple-input multiple-out (MIMO) JCAS system in terms of the resulting radar signal-to-interference-plus-noise ratio (SINR), using the signal model  in \cite{strum13}.  In \cite{liumu}, the authors used separated antenna arrays to realize dual-function JCAS systems. A multi-objective function was further applied to trade off the similarity of the generated waveform to the desired one in \cite{liuweight}. It is noted that all these papers use a co-located DFRC transceiver, similar to a mono-static radar, and face an essential requirement of the full-duplex capability of the transceiver \cite{fullduplex}. Alternative solutions other than full-duplex transceiver exist but require changes to existing network infrastructure, e.g.,  the authors in \cite{niICC} used a synchronized single-antenna sensing receiver that is sufficiently separated from the transmitter.

Since the full-duplex technology is not quite mature, there exists an optional transceiver setup for realizing JCAS in PMNs, similar to a bi-static radar \cite{passive10}, where the sensing receiver is physically separated from the transmitter. This setup is consistent with the uplink sensing as defined in \cite{lushan}. Such a setup can be implemented with minimal network changes only and is a favorite option in the near term.
Some works have been done for realizing JCAS for the uplink channels \cite{andrea20,ICC20ws}.
In \cite{andrea20}, the authors investigated the uplink sensing in a 5G cellular network using  massive MIMO and coexisting with a radar in the same frequency band. The authors adopted broadband OFDM modulation and obtained the uplink channel via minimum-mean-squared-error (MMSE) or zero-forcing (ZF) based processing schemes. In \cite{ICC20ws}, the authors proposed a receiver architecture for DFRC systems and obtained its corresponding uplink communication channel capacity and radar channel capacity, respectively.
The main challenges for realizing this setup in PMNs are (1) the unavailability of clock-level synchronization between the sensing receiver and the transmitter; (2) the relatively low angle-of-arrival (AoA) estimation accuracy due to the rich multi-path environment in mobile networks \cite{lushan}.
Perfect synchronization was assumed in most recent papers about JCAS, whereas the asynchrony between the sensing receiver and the transmitter is not addressed yet. Some papers on cognitive radio have dealt with the asynchronous issues \cite{asyofdm,asyofdm16} but they are unrelated to JCAS. These works analyzed the interferences caused by asynchrony but did not provide an effective way for parameter estimation.
 Without clock-level synchronization between the sensing receiver and the transmitter, both timing offsets (TOs) and carrier frequency offsets (CFOs) can occur \cite{IndoTrack,widar2.0}, leading to sensing ambiguity and degraded accuracy in estimating delay and Doppler frequency of targets.

To handle the asynchronous transceivers, a limited number of works on passive WiFi sensing have been proposed based on a cross-antenna cross-correlation (CACC) method \cite{IndoTrack,widar2.0,asyJCAS}. The underlying principle of CACC is that TOs can be removed by computing the cross-correlation between signals of multiple receiving antennas and exploiting the same TO across multiple antennas in one device. In \cite{IndoTrack}, CACC is applied to resolve the AoA estimation problem for device-free human tracking with commodity WiFi devices. In \cite{widar2.0}, CACC is used to resolve the ranging estimation problem for passive human tracking using a single WiFi link. Unfortunately, there exists a derivative problem with the CACC method, i.e., the outputs from CACC contain mirrored unknown parameters. The mirrored parameters double the number of unknown parameters and also obscured the sign of Doppler frequencies, leading to a degraded sensing accuracy. The author in \cite{IndoTrack} proposed an add-minus suppression (AMS) method that suppresses the mirrored parameters and extracts the actual ones. However, the AMS method needs the power of static paths to be much stronger than that of the dynamic paths, otherwise the mirrored component is suppressed slightly in a rich multi-path environment.

To overcome the challenge of low AoA estimation accuracy in the physically-separated transceiver, techniques based on spatial smoothing and combining measurements in spatial and other domains have been proposed \cite{spotFi,high_reso,niTCOM19}. In \cite{spotFi}, the authors proposed a high-resolution two-dimension (2D) MUSIC estimator via spatial smoothing, to obtain accurate AoA estimates with a small scale of antennas.  In \cite{high_reso}, the authors jointly combined spatial and temporal measurements to obtain high-resolution AoA estimation. In \cite{niTCOM19}, the authors defined a spatial path filter that is used to separate signals from multiple propagation paths and obtained the AoAs via CACC outputs. All these methods are designed for narrowband systems and they do not have the capability of jointly estimating AoAs and other parameters such as delay and Doppler frequency.

In this paper, we propose a broadband uplink sensing scheme for PMNs with physically-separated asynchronous transceivers and OFDM modulation. There are two key novelties in our scheme. Firstly, since the mirrored components derived from the CACC outputs have not been addressed properly in prior works, we propose a mirrored-MUSIC algorithm that jointly processing the actual and mirrored parameters to overcome the mirrored parameter problem associated with CACC. Secondly, noting that the estimates of AoAs are not included in the proposed mirrored-MUSIC algorithm, we propose a high-resolution AoA estimation algorithm that can obtain high-resolution estimates of AoAs by combining spatial and other domain measurements. Our proposed scheme is applicable for practical scenarios requiring a single static user equipment (UE) and a line-of-sight (LOS) path between the static transmitter and the base station (BS). Our major contributions are summarized as follows.

\begin{itemize}
\item We provide a practical radar sensing scheme that can be implemented in mobile networks without the requirement of clock-level synchronization between the transmitter and receiver. By using a CACC method to mitigate the TOs and CFOs resulted from asynchrony, our scheme relaxes stringent clock-level synchronization between physically-separated transceivers. The mean-squared-error (MSE) between the CACC output and our desired signal also keeps at a minimum level.
\item We propose a mirrored-MUSIC algorithm to handle the mirrored outputs of the CACC method and a general problem where the test basis vectors show mirror symmetry. The algorithm, which halves the unknown variables without introducing any approximation, has lower complexity and better performances compared to conventional MUSIC. This algorithm can also be applied to many other applications, such as traditional harmonic retrieval problems with the sinusoidal modulated signals \cite{rHRP04,rHRP05}, ESPRIT \cite{high_reso}, and  the matrix pencil method \cite{2013reducing}.
\item We develop a high-resolution MUSIC-based AoA estimation algorithm that combines measurements from spatial, temporal, and frequency domains. This algorithm equivalently increases the samples in the spatial domain and hence significantly improves the resolution of AoAs, compared with estimating the AoAs in the spatial domain only. In particular, a new issue of ``ambiguity of basis vectors'' occurs when integrating multiple parameters into one domain measurement. Our algorithm resolves this critical issue in combining multi-domain measurements by selecting multiple peaks from the MUSIC outputs.
\end{itemize}

The rest of this paper is organized as follows. Section \ref{sec-system} introduces system and channel models. Section \ref{domi} presents the CACC method. Section \ref{Mirror} introduces the proposed mirrored-MUSIC algorithm for estimating the actual delays and Doppler frequencies. Section \ref{HAoA} presents the high-resolution AoA estimation scheme. Extensive simulation results are presented in Section \ref{SIM}. Finally, conclusions are drawn in Section \ref{sec-conc}.

Notations: $\rm\bf a$ denotes a vector, $\rm\bf A$ denotes a matrix, italic English letters like $N$ and lower-case Greek letters $\alpha$ are a scalar, $\angle a$ is the phase angle of complex value $a$. $|{\rm\bf A }|, {\rm\bf A }^T, {\rm\bf A }^H, {\rm\bf A }^\dag$ represent determinant value, transpose, conjugate transpose, pseudo inverse, respectively. We denote Frobenius norm of a matrix as $\|{\rm\bf A }\|_F$. We use ${\rm diag}(\alpha_1,\cdots,\alpha_k)$ to denote a diagonal matrix. $[{\rm\bf A }]_{ N}$ is the $N$th column of a matrix  and $[{\rm\bf A }]^{ N}$ is the $N$th row of a matrix.

\begin{figure}[t]
	\centering
	\includegraphics[width=0.7\columnwidth]{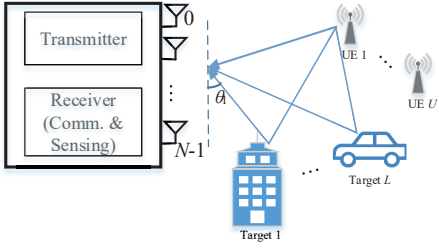}
	\caption{Illustration of the system model for uplink sensing.}
	\label{system}
\end{figure}

\section{System and Channel Models}\label{sec-system}
We consider the uplink communication and sensing in a PMN, as shown in Fig. \ref{system}. Multiple UEs communicate with a BS. The BS is physically static and uses received uplink signals for both communication and sensing. Each UE has one antenna and the BS has a limited number of $N$ antennas. Our proposed scheme in this paper requires the following setups:
\begin{itemize}
	\item The signals used for sensing are from a specific UE of which location is fixed and known to the BS.
	\item There is a LOS propagation path between the BS and the UE used for sensing. The power of the LOS path is much larger than that of non-LOS (NLOS) paths.
\end{itemize}
This setup is practically feasible for PMNs. The fixed UE can be a node that provides fixed broadband access in the mobile network. We can adopt the millimeter-wave frequency band to guarantee the dominating power of the LOS path.
In PMN, several types of signals may be used for sensing. Referring to the fifth-generation (5G) mobile network, they can be demodulation reference signals (DMRSs) that are specifically provided for channel estimation, synchronization signal blocks (SSBs), and even demodulated data symbols  \cite{DMRSliu,SSBpatent,ykls1,an2020}. Without loss of generality, we consider sensing via the uplink signal from a specific UE, denoted as UE 1.

At all UEs, we adopt a simplified packet structure, as shown in Fig. \ref{fig-sig}. In each packet, training symbols, denoted as preambles, are followed by a sequence of data symbols. OFDM modulation is applied across the whole packet. These data symbols can be empty if the packet is a DMRS. We only use the OFDM preambles for sensing. The preambles can also be used for synchronization and channel estimation for communications, which needs different processing at the BS. In this paper, we would like to use the preambles for sensing multiple targets in the PMN. The parameters of targets including the propagation delay, Doppler frequency, and AoA need to be obtained. Our proposed scheme can also be applied to other systems with similar signal structure, such as WiFi systems.

\begin{figure}[t]
	\centering
	\includegraphics[width=0.9\columnwidth]{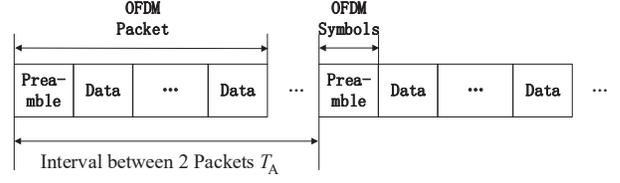}
	\caption{Illustration of transmitted OFDM packets at the UE baseband.}
	\label{fig-sig}
\end{figure}

Without losing generality, we assume each packet has only one preamble. For both preamble symbol and data symbol, each of them has $G$ subcarriers with a subcarrier interval of $1/T$, where $T$ denotes the length of an OFDM symbol. Each of the OFDM symbols is prepended by a cyclic prefix (CP) of period $T_{\rm C}$. Our scheme works if and only if a segment of subcarriers with an interleaved interval is available for UE 1. When multiple UEs communicate with the BS, each UE occupies a unique segment of subcarriers with the interleaved interval as in \cite{sit}. For notational simplicity, we assume that UE 1 occupies the whole preamble symbol here. Mathematically, the $m$th preamble symbol can be expressed as \cite{gudelay,strum13}
\begin{align}
s(t|m)= \sum\limits_{g=0}^{G-1} \exp\left(j2\pi g\frac tT\right){\rm rect}\left(\frac{t}{T+T_{\rm C}}\right)  x[m,g],
\end{align}
where $x[m, g]$ is a modulated symbol transmitted on the $g$th subcarrier of the $m$th preamble symbol and ${\rm rect}\left(\frac{t}{T+T_{\rm C}}\right)$ denotes a rectangular window of length $T+T_{\rm C}$.

The BS receives the preambles using a uniform linear array (ULA) of $N$ antennas.  The uplink channel between receiver at BS and the transmitter at UE 1 has $L$ NLOS paths reflected or refracted from $L$ targets, together with a dominating LOS path,  where the index of the LOS path is denoted as $l=0$. Let $\alpha_l$, $f\dop$, $\tau_l$ and $\theta_l$ denote the channel gain, the Doppler frequency, the propagation delay, and the AoA of the $l$th path, respectively. Due to the fixed locations of BS and UE 1, we assume that the parameters, $\tau_0$ and $\theta_0$, which correspond to the LOS path, are known at the BS, and $f\dopz$ is $0$. We also assume that $|\alpha_0|\gg|\alpha_l|, \forall l\in\{1,\cdots, L\}$. Note that the Doppler frequency of the $l$th path comes from the $l$th target of the channel, which can be either positive or negative depending on the moving directions.

We assume that $M$ packets are sent at the same interval, denoted as $T_{\rm A}$, at the UE baseband. Since there is typically no synchronization at clock level between BS and UE 1, the received signal has an unknown time-varying TO, denoted as $\delta_\tau(m)$, associated with the clock asynchrony, even if the packet level synchronization is achieved. Hence, the total time delay during signal propagation for the $l$th target as seen by BS equals  $\tau_l + \delta_\tau(m)$. In \cite{widar2.0}, it is shown that there also exists an unknown time-varying CFO due to the asynchronous carrier frequency, denoted as $\delta_f(m)$. The received time-domain signal corresponding to the preamble symbol in the $m$th packet can be represented as \cite{widar2.0}
\begin{align}\label{rt}
{\bf y}(t|m)=
&\sum\limits_{l=0}^{L}\alpha_le^{j2\pi m(T_{\rm A}+\delta_\tau(m)+\tau_l)(f\dop+\delta_f(m)) } \times\notag\\
&s(t-\tau_l-\delta_\tau(m)) {\bf a}(\p_l)+{\bf z}(t|m)\notag\\
\approx&\sum\limits_{l=0}^{L}\alpha_le^{j2\pi mT_{\rm A}(f\dop+\delta_f(m)) }\times\notag\\ &s(t-\tau_l-\delta_\tau(m)) {\bf a}(\p_l)+{\bf z}(t|m),
\end{align}
where the vector, ${\bf a}(\p_l)=\exp[j\p_l(0,1,\cdots, N-1)]^T$, is the array response vector of size $N\times 1$, with $\p_l$ being $\frac{2\pi d}{\lambda}\cos\theta_l$, $d$ denoting the antenna interval, $\lambda$  denoting the wavelength, and $\theta_l$ being the AoA from the $l$th target, and ${\bf z}(t|m)$ is a complex additive-white-Gaussian-noise (AWGN) vector with zero mean and variance of $\sigma^2$. TO is typically time-varying and has a random value that changes during any two discontinuous transmissions. The CFO may slowly vary over time. It is noted that TO and CFO are mixed with the actual propagation delay and the actual Doppler frequency, respectively. Hence they can directly cause ambiguity of ranging and velocity measurements. They also make the total delay and total Doppler frequency vary with time and prevent from aggregating signals for joint processing. It should be noted that, for the communication purpose, there is no need to distinguish the actual parameters with these offsets, since they can be estimated as a whole value and then be removed.  As for the radar sensing purpose,  these offsets have to be mitigated since  the range and the velocity of targets only depend on actual parameters.
Note that we use the approximation $e^{j2\pi m(T_{\rm A}+\delta_\tau(m)+\tau_l)(f\dop+\delta_f(m)) }\approx e^{j2\pi mT_{\rm A}(f\dop+\delta_f(m))}$, since the timing values of $(\delta_\tau(m)+\tau_l)$ are much smaller than $T_{\rm A}$ and $(f\dop+\delta_f(m))$ is also small in relation to the sampling rate.

After removing CP from the received time-domain signal, we then transform the signal into frequency domain via  $G$-point fast-Fourier-transform (FFT)'s.  Referring to \eqref{rt}, the received frequency-domain signal is
\begin{align}
y_n[m,g]=&\sum\limits_{l=0}^{L}\alpha_l e^{jn\p_l} e^{j2\pi mT_{\rm A}(f\dop+\delta_f(m))}
\times\notag\\
&e^{-j2\pi \frac{g}{T}(\tau_l+\delta_\tau(m))}x[m,g]+z_n[m,g],
\end{align}
where $y_n[m,g]$ is the received frequency-domain signal on the $g$th subcarrier at the $n$th receiving antenna of the $m$th OFDM preamble symbol, and $z_n[m,g]$ is a complex AWGN with zero mean and variance of $\sigma^2$. The actual value of $|x[m,g]|^2$ has insignificant impact on our proposed scheme. For simplicity, we assume $|x[m,g]|^2=1$.

\section{CACC for Mitigating TOs and CFOs}\label{domi}

As we mentioned in Section II, the actual delays and the actual Doppler frequencies are mixed with TOs and CFOs, respectively. In this section, we adopt and extend the CACC method to generate signals with TOs and CFOs being removed, in order to obtain the delay and the Doppler frequency of targets.

We decompose the received signals into three parts, i.e.,
\begin{align}
 y_n[m,g]=D_n[m,g]+I_n[m,g]+z_n[m,g],
\end{align}
where $D_n[m,g]$ denotes the received signal from the LOS path, given by
\begin{align}
D_n[m,g]=& \alpha_0 e^{jn\p_0}e^{j2\pi mT_{\rm A}\delta_f(m)}
 e^{-j2\pi \frac{g}{T}(\tau_0+\delta_\tau(m))}x[m,g],
\end{align}
and $I_n[m,g]$ is the received signals reflected or refracted from the targets, given by
\begin{align}
&I_n[m,g]\notag\\
=&\sum\limits_{l=1}^{L} \alpha_l e^{jn\p_l}e^{j2\pi mT_{\rm A}(f\dop+\delta_f(m))} 
e^{-j2\pi \frac{g}{T}(\tau_l+\delta_\tau(m))}x[m,g].
\end{align}

The CACC operation makes it possible to mitigate both TO and CFO. This operation computes the cross-correlation between different antennas and is generally used in estimating the AoA with hybrid subarrays \cite{niTCOM19}.
We select one antenna that has the largest received average power as \textit{reference antenna}. Without losing generality, we assume that the reference antenna is the $0$th antenna. Neglecting the noise term, the CACC operation between the $n$th antenna and the $0$th antenna  generates
\begin{align}
\rho_n[m,g]&=y_n[m,g]y_0^H[m,g]\notag\\
           &\approx(D_n[m,g]+I_n[m,g])(D_0^H[m,g]+I_0^H[m,g])\notag\\
           &\triangleq\rho_n^{(1)} +\rho_n^{(2)}[m,g]+\rho_n^{(3)}[m,g]+\rho_n^{(4)}[m,g],
\end{align}
where $y_0[m,g]$ is the received signal at the reference antenna,
$\rho_n^{(1)}[m,g] =D_n[m,g]D_0^H[m,g]$,
$\rho_n^{(2)}[m,g]=I_n[m,g]I_0^H[m,g]$, $\rho_n^{(3)}[m,g]=D_n[m,g]I_0^H[m,g]$, and $\rho_n^{(4)}[m,g]=I_n[m,g]D_0^H[m,g]$.
 For those CACC outputs, we have the following proposition.
\begin{proposition}\label{T1}
In $\rho_n[m,g]$, $\rho_n^{(1)}[m,g]$ is invariant with $m$ and $g$. The 2D-FFT output of  $\rho_n^{(3)}[m,g]+\rho_n^{(4)}[m,g]$ over $m$ and $g$ shows an impulsive shape that is centred around $\tau_l$ and $f\dop$. The power of $\rho_n^{(2)}[m,g]$ is significantly lower than the other three terms.
\end{proposition}
\begin{proof}
The proof is provided in Appendix \ref{T1proof}.
\end{proof}
According to Proposition \ref{T1}, by using a 2D high-pass filter with respect to $m$ and $g$, we can remove the invariant component from the CACC outputs and obtain $\hat\xi_n[m,g]\approx\rho_n^{(3)}[m,g]+\rho_n^{(4)}[m,g]$. The cut-off frequency of the 2D high-pass filter depends on the bandwidth of $\hat\xi_n[m,g]$. From Appendix  \ref{T1proof}, the cut-off frequency of $\hat\xi_n[m,g]$ is $(\omega_f,\omega_\tau)=\left(\min\left|\pi T_{\rm A}f\dop\right|,\min(\pi(\tau_l-\tau_0)/T)\right)$.

From the expression of $\rho_n^{(3)}[m,g]+\rho_n^{(4)}[m,g]$, we note that  $\hat\xi_n[m,g]$ contains the actual sensing parameters without TOs and CFOs. The output from the high-pass filter can be represented as
\begin{align}\label{cros}
  \hat\xi_n[m,g]=  &  \rho_n[m,g]-\bar\rho_n \notag\\
  \approx& \rho_n^{(3)}[m,g]+\rho_n^{(4)}[m,g]\notag\\
  = & \sum\limits_{l=1}^L\alpha_0\alpha_l^He^{j2\pi mT_{\rm A}(-f\dop)}e^{-j\frac{2\pi g}{T}(\tau_0-\tau_l)}e^{jn\p_0}  +\notag\\
   &\sum\limits_{l=1}^{L}\alpha_l\alpha_0^He^{j2\pi mT_{\rm A}f\dop}e^{-j\frac{2\pi g}{T}(\tau_l-\tau_0)}e^{jn\p_l}\notag\\
  \triangleq &\xi_n[m,g],
\end{align}
where $\bar\rho_n$ is the low-pass component in $\rho_n[m,g]$.
Note that the delays in \eqref{cros} become relative values, i.e., $\tau_l-\tau_0$, $(0<\tau_0<\tau_l)$. Since $\tau_0$ is assumed to be known at the BS, the problem of estimating $\tau_l$ becomes how to estimate the relative delays.  The error between $\hat\xi_n[m,g]$ and $\xi_n[m,g]$ is resulted from $\rho_n^{(2)}[m,g]$ and hence not affected by the noise. We define this error as the \textit{input error}, i.e. $e^2=\mathbb E\|\hat\xi_n[m,g]-\xi_n[m,g]\|^2$.

In $\xi_n[m,g]$, $\rho_n^{(3)}[m,g]$ can be seen as the \textit{side product}, which is mixed with the \textit{actual component} of interest,  $\rho_n^{(4)}[m,g]$, since $\rho_n^{(4)}[m,g]$ already contains all sensing parameters. It would be redundant to estimate $\rho_n^{(3)}[m,g]$ and the estimation would require a doubled number of samples if conventional methods are used to estimate the actual component and the side product together.
One idea is to separate $\rho_n^{(3)}[m,g]$ from $\rho_n^{(4)}[m,g]$, which is a challenging task. The AMS method was proposed in \cite{widar2.0} and \cite{IndoTrack} to remove the side product. However, this method does not always work, particularly when the number  of   signal propagation paths is large in a rich multi-path propagation environment. In Appendix \ref{MSEofXI}, we briefly describe the AMS method and compare its input error with $\xi_n[m,g]$ that is adopted in our scheme. We show that our adopted $\xi_n[m,g]$ has a smaller input error than that of the AMS method.

\section{Mirrored-MUSIC for Estimating Propagation Delays and Doppler Frequencies}\label{Mirror}

In this section, we propose a mirrored-MUSIC algorithm that is tailored to directly estimating conjugated variables from signals similar to the one in \eqref{cros}.

MUSIC-based algorithms have been widely used for estimating different parameters of channels, including delay, Doppler frequency, and AoA \cite{passive10,widar2.0,high_reso}. With a given signal matrix, conventional MUSIC finds the formulation of basis vectors of the signal matrix. Utilizing the fact that the basis vectors fall into the null-space of the signal matrix, the parameter can be obtained by checking if the candidate basis vector with a testing parameter falls into the null-space of the signal matrix. Conventional MUSIC would construct $\xi_n[m,g]$ into a matrix. The matrix corresponding to $\xi_n[m,g]$ has a doubled number of parameters to be estimated.
We will exploit this redundancy, and the estimation of both delay and Doppler frequency will be transformed into an equivalent problem with halved unknown variables. This can achieve significantly improved performance compared to the conventional MUSIC algorithms.
\subsection{Proposed Mirrored-MUSIC Algorithm}
With integrating the side product and the actual component, we rewrite $\xi_n[m,g]$ to a general expression as
\begin{align}\label{simp}
\xi_n[m,g]=
\sum\limits_{l'=-L, l'\neq 0}^{L}P_{l'} e^{j m{\bar f}_{{\rm D},l'}}e^{-jg\bar\tau_{l'}} e^{jn\p_{l'}}.
\end{align}
In \eqref{simp}, the variables with indexes $l'<0$ represent those actual ones to be estimated and those with $l'>0$ belong to the side product, that is, $P_{l'}$ equals $\alpha_0\alpha_l^H $ and $\alpha_l\alpha_0^H$ when $l'>0$ and $l'<0$, respectively, ${\bar f}_{{\rm D},l'}=2\pi T_{\rm A}f\dop(-1)^{{\rm step}(l')}$ denotes the mirrored Doppler frequency,  $\bar\tau_{l'}=\frac{2\pi}T(\tau_0-\tau_l)(-1)^{{\rm step}(l')}$ denotes the mirrored delay, and $\p_{l'}$ equals $\p_0$ and $\p_l$ when $l'>0$ and $l'<0$, respectively. It is worth pointing out that only the terms of Doppler frequency and delay exhibit \textit{mirror symmetry}, i.e., the Doppler frequency and the delay of the side product are opposite to those of the actual component. The AoA terms in the side product have no such a property. By exploiting the mirror symmetry of both Doppler frequency and delay, we can reduce the number of estimates for delay  and Doppler frequency by half to $L$, respectively.

 Let us generate two types of mirrored signal vectors based on $\xi_n[m,g]$ by adding a vector with its reversed version. The two new vectors are given by
\begin{align}\label{pp}
 {\bf p}_n[m,g]
=& [\xi_n[m,g], \cdots, \xi_n[m+P,g] ]^T +\notag\\
&[\xi_n[m+P,g], \cdots, \xi_n[m,g] ]^T,
\end{align}
and
\begin{align}\label{qq}
 {\bf q}_n[m,g]
=&[\xi_n[m,g], \cdots, \xi_n[m,g+Q]]^T +\notag\\
&[\xi_n[m,g+Q], \cdots, \xi_n[m,g] ]^T,
\end{align}
where $P$ and $Q$ satisfy $L\leq P<M-L$ and $L\leq Q<G-L$, respectively, $P\in {\mathbb N}, Q\in {\mathbb N}$. The applied ranges of $P$ and $Q$ guarantee that there are at least $L$ mirrored vectors that are linearly independent of each other. The proposed mirrored vectors exhibit mirror symmetry too. It is clear to see that the $i$th entry of ${\bf p}[m,g]$ is the same with the $(P-i)$th entry and the $i$th entry of ${\bf q}[m,g]$ is the same with the $(Q-i)$th entry.
More importantly, we have the following theorem.
\begin{theorem}
The signal vectors, ${\bf p}_n[m,g]$ and ${\bf q}_n[m,g]$,  have only $L$ basis vectors, respectively.
\end{theorem}
\begin{proof}
We define two types of mirrored basis vectors as
\begin{align}\label{basis1}
{\bf p}_{m}(\bar f_{{\rm D},l'}) =&\left[e^{jm \bar f_{{\rm D},l'} }, \cdots, e^{j(m+P) \bar f_{{\rm D},l'} }\right]^T  +\notag\\
&\left[e^{j(m+P) \bar f_{{\rm D},l'} }, \cdots, e^{jm \bar f_{{\rm D},l'} } \right]^T,\notag\\
&l'\in\{\pm1,\cdots,\pm L\}, m\in\{0,\cdots,M-P-1\},
\end{align}
and
\begin{align}\label{basis2}
{\bf q}_g(\bar\tau_{l'}) =&\left[e^{-jg\bar\tau_{l'} }, \cdots, e^{-j(g+Q) \bar\tau_{l'}}       \right]^T +\notag\\
&\left[e^{-j(g+Q) \bar\tau_{l'}}, \cdots, e^{-jg \bar\tau_{l'} } \right]^T,\notag\\
&l'\in\{\pm1,\cdots,\pm L\}, g\in\{0,\cdots,G-Q-1\}.
\end{align}
It is noted that
\begin{align}
{\bf p}_n [m,g]=\sum\limits_{l'=-L,l'\neq0}^{  L}P_{l'}e^{-jg\bar\tau_{l'}}e^{jn\p_{l'}}{\bf p}_m(\bar f_{{\rm D},l'}),
\end{align}
and
\begin{align}
{\bf q}_n [m,g]   =\sum_{l'=-L,l'\neq0}^{ L}P_{l'}e^{jm\bar f_{{\rm D},l'}}e^{jn\p_{l'}}{\bf q}_g(\bar\tau_{l'}).
\end{align}
Hence, ${\bf p}_{m}(\bar f_{{\rm D},l'})$ and ${\bf q}_g(\bar\tau_{l'})$, $l'\in\{\pm 1, \cdots, \pm L\},$ are the basis vectors of ${\bf p}_n[m,g]$ and ${\bf q}_n[m,g]$, respectively.

Then, we need to prove that those mirrored basis vectors only span a space of rank $L$. For ${\bf p}_m(\bar f_{{\rm D},l'})$, we have
\begin{align}
{\bf p}_m(\bar f_{{\rm D},-l'}) =&\left[e^{-jm \bar f_{{\rm D},l'} }, \cdots, e^{-j(m+P) \bar f_{{\rm D},l'} } \right]^T+\notag\\
&\left[e^{-j(m+P) \bar f_{{\rm D},l'} }, \cdots, e^{-jm \bar f_{{\rm D},l'} } \right]^T\notag\\
            =& \left[e^{jm \bar f_{{\rm D},l'} }, \cdots, e^{j(m+P) \bar f_{{\rm D},l'} } \right]^Te^{-j(2m+P) \bar f_{{\rm D},l'}} +\notag\\
            &\left[e^{j(m+P) \bar f_{{\rm D},l'} }, \cdots, e^{jm \bar f_{{\rm D},l'} } \right]^Te^{-j(2m+P) \bar f_{{\rm D},l'}}\notag\\
            =&{\bf p}_m(\bar f_{{\rm D},l'}) e^{-j(2m+P) \bar f_{{\rm D},l'}}.
\end{align}
and
\begin{align}
{\bf p}_{m+1}(\bar f_{{\rm D},l'}) =& [e^{j(m+1) \bar f_{{\rm D},l'} }, \cdots, e^{j(m+1+P) \bar f_{{\rm D},l'} } ]^T+\notag\\
&[e^{j(m+1+P) \bar f_{{\rm D},l'} }, \cdots, e^{j(m+1) \bar f_{{\rm D},l'} } ]^T\notag\\
            &={\bf p}_m(\bar f_{{\rm D},l'})e^{ j \bar f_{{\rm D},l'}}.
\end{align}
Therefore, all vectors of ${\bf p}_m(\bar f_{{\rm D},-l'})$, $\forall m, \forall l'$, have only $L$ linearly independent basis vectors, i.e.,
$\{{\bf p}_0(\bar f_{{\rm D},l'})\}_{l'=1}^L$.
Likewise, for ${\bf q}_g(\bar\tau_{l'})$, we have
\begin{align}
{\bf q}_g(\bar\tau_{-l'})={\bf q}_{g,l'}e^{j(2g+Q)\bar\tau_{l'}},
\end{align}
and
\begin{align}
{\bf q}_{g+1}(\bar\tau_{l'})={\bf q}_{g,l'}e^{-j \bar\tau_{l'}}.
\end{align}
Therefore,  all vectors of ${\bf q}_g(\bar\tau_{l'})$, $\forall g, \forall l'$, have only $L$ linearly independent basis vectors, i.e., $\{ {\bf q}_g(\bar\tau_{l'})\}_{l'=1}^L$.
Overall,  $2L(M-P)$ ${\bf p}_{m,l'}$ span a space of rank $L$ and $2L(G-Q)$ ${\bf q}_{g,l'}$ span  a space of rank $L$.
\end{proof}

The class of MUSIC algorithms requires the formulation of basis vectors that can span the entire signal space. Directly using $\xi_n[m,g]$ to construct the signal matrix, as in the case of the conventional MUSIC, would require $2L$ basis vectors due to the side product. The proposed mirrored signal vectors have only $L$ basis vectors that can span the whole signal space.

For  ${\bf p}_n[m,g]$, $n\in\{1,\cdots,N-1\}$, $m\in\{0,\cdots, M-P-1\}$, and $g\in\{0,\cdots, G-1\}$, there are $G(N-1)(M-P)$ vectors in total. For  ${\bf q}_n[m,g]$, $n\in\{1,\cdots,N-1\}$, $m\in\{0,\cdots, M-1\}$, and $g\in\{0,\cdots, G-Q-1\}$, there are $M(N-1)(G-Q)$ vectors in total.
Due to the high computational complexity of singular value decomposition (SVD) in MUSIC, stacking all signal vectors into a matrix would lead to prohibitive complexity. Instead, we fix $n$ and $g$ as $n_0$ and $g_0$, respectively, and stack ${\bf p}_{n_0} [m,g_0]$, $m\in\{0,\cdots,M-P-1\}$, into a matrix of dimension $(P+1)\times (M-P)$, i.e.,
\begin{align}
{\bf P}=[{\bf p}_{n_0}[0,g_0], {\bf p}_{n_0}[1,g_0],\cdots, {\bf p}_{n_0} [M-P-1,g_0] ].
\end{align}
Neglecting the noise term, all column vectors in ${\bf P}$ can be expressed by $\left\{{\bf p}_0(\bar f_{{\rm D},l'})\right\}_{l'=1}^L$. Hence, the rank of ${\bf P}$ is $L$. Similarly, we can stack all ${\bf q}_{n_0} [m_0,g]$, $g\in\{0,\cdots,(G-Q-1)\}$, into a matrix, denoted as ${\bf Q}$, which is also of rank $L$. The optimal value of $n_0$ is demonstrated in Proposition 2. As for optimizing $m_0$ and $g_0$, we can use the signals with the largest received power on average.

The matrix ${\bf P}$ is only related to the Doppler frequency.  Let us perform the SVD of ${\bf P}$, i.e., ${\bf P}={\bf U}_{\rm P}{\bf E}_{\rm P}{\bf V}_{\rm P}^H$, where ${\bf E}_{\rm P}$ is an $L\times L$ diagonal matrix, ${\bf U}_{\rm P}$ is the left singular matrix, and ${\bf V}_{\rm P}$ is the right singular matrix. Denoting the null-space of ${\bf U}_{\rm P}$ as $\bar {\bf U}_{\rm P}$, we can then estimate the Doppler frequency via
\begin{align}\label{music1}
 {\rm Peak}^L\left(\frac1{\left\| {\bf p}^H_0( 2\pi T_{\rm A} f') \bar{\bf U}_{\rm P}\right\|_F^2}\right),
\end{align}
where ${\rm Peak}^L(\cdot)$ denotes the operation that takes $L$ estimates corresponding to the $L$ largest peak values of the function in the bracket, ${\bf p}_0( 2\pi T_{\rm A} f')$ has the same expression as \eqref{basis1} with $m=0$, and $f', f'\in\left(0,\frac{1}{T_{\rm A}}\right),$ is a quantized Doppler frequency for testing, with the interval between adjacent $f'$ being $\frac{1}{T_{\rm A}(P+1)}$. Note that $f'$ only needs to be tested from $0$ to $\frac{1}{T_{\rm A}}$ due to the mirror symmetry of the proposed basis vectors. Hence, the proposed mirrored-MUSIC algorithm can increase the estimation accuracy with the same quantization interval.

For ${\bf Q}={\bf U}_{\rm Q}{\bf E}_{\rm Q}{\bf V}^H_{\rm Q}$, denoting the null-space of ${\bf U}_{\rm Q}$ as $\bar{\bf U}_{\rm Q}$, we can estimate the relative delay, in parallel with estimating the Doppler frequency, via
\begin{align}\label{music2}
 {\rm Peak}^L\left( \frac1{\left\|{\bf q}_0^H\left(\frac{2\pi\tau'}T\right)\bar{\bf U}_{\rm Q}\right\|_F^2}\right),
\end{align}
where ${\bf q}_0 \left(\frac{2\pi\tau'}T\right)$ has the same expression as \eqref{basis2}, and $\tau'$, $\tau'\in(0, T)$, is a quantized delay for testing, with the interval between adjacent $\tau'$ being $\frac{T}{(Q+1)}$.
Similar to $f'$, $\tau'$ only needs to be tested from $0$ to $T$ due to the mirror symmetry.

To determine the number of targets, we can adopt existing algorithms such as the well-known minimum description length (MDL) method and the simplified one in \cite{MDL}, using the diagonal elements in the singular value matrix of both ${\bf P}$ and ${\bf Q}$.

\subsection{Pair Matching and Doppler Frequency's Sign Determination}

There are two problems yet to be solved, following the 'Peak' function in \eqref{music1} and \eqref{music2}. Firstly, the estimates of delays and Doppler frequencies are not automatically matched to one target. Hence, we need to make a pair for each estimate of delay with each estimate of Doppler frequency. Secondly, the sign of the estimate of Doppler frequency is yet to be determined since we only obtain the absolute values of Doppler frequency. The sign of delay does not need to be determined, since  $\bar\tau_l=\tau_l-\tau_0$ is larger than zero by default. Note that the two derivative problems also exist and they are even more challenging when a conventional MUSIC algorithm is applied, because conventional MUSIC would obtain two values for one Doppler frequency and it needs to determine the correct value from two individual estimates. We now solve these two problems based on the CACC outputs.

From \eqref{simp}, we see that the term of AoA does not have mirror symmetry. All AoAs of $\xi_n[m,g]$ equal $\p_0$ when $l'>0$. Since $\p_0$ is known to the sensing receiver (BS), we can utilize $\p_0$ to address the above-mentioned two problems. Due to the undetermined sign of Doppler frequency and the unmatched delay and Doppler frequency, there are $2L^2$  candidates, i.e.,
\begin{align}\label{pair}
 ({\hat f}_{{\rm D}, l_x}, \hat{\tau}_{l_y}), l_x\in\{\pm 1,\cdots, \pm L\}, l_y\in\{1,\cdots, L\},
\end{align}
where ${\hat f}_{{\rm D}, l_x}$ and $\hat\tau_{l_y}$ are candidates to be paired. The value of $\hat\tau_{l_y}$, which is larger than zero, is the estimate obtained from \eqref{music2}. The absolute value of ${\hat f}_{{\rm D}, l_x}$  is the estimate obtained from \eqref{music1}.
Note that there are two opposite values for one Doppler frequency  and only one of them matches the actual one. The actual one has the maximum combining gain in the following function that combines $\xi_n[m,g]$, i.e.,
\begin{align}\label{P_XI}
 &P_{\xi}(l_x,l_y)\notag\\
 =&\sum\limits_{m=0}^{M-1}\sum\limits_{g=0}^{G-1}\sum\limits_{n=0}^{N-1} \xi_n[m,g]e^{ jm2\pi T_{\rm A}\hat f_{{\rm D},l_x}-jg\frac{2\pi g}T(\hat{\tau}_{l_y}-\tau_0)-jn\p_0}.
\end{align}
We can  first select one  out of $2L^2$  candidates that maximizes the absolute value of  $P_{\xi}(l_x,l_y)$. Supposing that the selected index of the obtained pair is $(l_{x_0}, l_{y_0})$, we remove both this pair and its mirrored index, $(-l_{x_0}, l_{y_0})$ from the set of candidates. Meanwhile, the sign of Doppler frequency for $f_{{\rm D},l_{x_0}}$ is determined. After removing the pair of $(\pm l_{x_0}, l_{y_0})$, the number of candidates is reduced to $2(L-1)^2$, i.e., $l_x\in\{\pm 1,\cdots, \pm L\}, l_x \notin \{\pm l_{x_0}\}$, $l_y\in\{ 1,\cdots, L\}, l_y  \notin\{ l_{y_0}\}$. We then match the next pair of Doppler frequency and delay. Repeating the process $L$ times, we can  match Doppler frequency with delay and obtain the sign of Doppler frequency simultaneously.

The whole process of estimating Doppler frequency and delay using the CACC outputs is summarized in Algorithm \ref{Alg_1}.

\begin{algorithm}[t]
	\caption{Proposed Mirrored-MUSIC Estimation Algorithm}\label{Alg_1}
	\begin{algorithmic}[1]
		\STATE {\bf Input: } $\xi_n[m,g]$ and $\p_0$.
		\STATE {\bf Initialization:} $P\in {\mathbb N}$ and $Q\in {\mathbb N}$ satisfy $L\leq P<M-L$ and $L\leq Q<G-L$.  Candidates of quantized Doppler frequency and delay for testing are selected uniformly over $(0,\frac1{T_{\rm A}})$ and $(0, T)$, respectively.
		
		\STATE Generate ${\bf p}_n [m,g]$ and ${\bf q}_n [m,g]$ according to \eqref{pp} and \eqref{qq}.
		
		\STATE Assemble ${\bf p}_{n_0}[m,g_0]$ from $m=0$ to $m=(M-P-1)$ into matrix $\bf P$. Assemble ${\bf q}_{n_0}[m_0,g]$ from $g=0$ to $g=(G-Q-1)$ into matrix $\bf Q$.
		
		\STATE SVD: ${\bf P}={\bf U}_{\rm P}{\bf E}_{\rm P}{\bf V}^H_{\rm P}$ and  ${\bf Q}={\bf U}_{\rm Q}{\bf E}_{\rm Q}{\bf V}^H_{\rm Q}$.
		\STATE Denote the null-space of ${\bf U}_{\rm P}$ and ${\bf U}_{\rm Q}$ as $\bar{\bf U}_{\rm P}$ and $\bar{\bf U}_{\rm Q}$, respectively.
		
		\STATE Estimate $\{f\dop\}_{l=1}^L$ and $\{\tau_l\}_{l=1}^L$ via \eqref{music1} and \eqref{music2}, respectively. The estimates are denoted as $ \hat{f}\dop$ and $\tau_l$.
		
		\STATE Generate $2L^2$ candidates for Doppler frequency and delay according to \eqref{pair}, and obtain  $2L^2$ $P_\xi (l_x,l_y)$ according to \eqref{P_XI}.
       \FOR{$X=L:-1:1$}
            \STATE Select one out of $2X^2$ candidates with the maximal $|P_\xi (l_x,l_y)|$, with the selected index being $i_X$.
            \STATE Find $l_x $ and $l_y $ that correspond to $i_X$.
            \STATE Remove  $P_\xi (l_x,\cdot)$ and $P_\xi (\cdot,l_y)$.
       \ENDFOR
		
		\STATE {\bf Output:}   $\hat {f}\dop$ and $\hat {\tau}_l$.
	\end{algorithmic}
\end{algorithm}

\subsection{Performance Analysis}\label{Accuracy}
Since the MUSIC-based estimators are non-linear approaches, we analyze the performance of the proposed mirrored-MUSIC using the perturbation methods as in \cite{MUSICaly,hiso25}. The mirrored-MUSIC is based on CACC, which makes the analysis more challenging than that in \cite{MUSICaly,hiso25}. Without losing generality, we analyze the performance for estimating Doppler frequency, and the performance for estimating delay can be similarly derived.

We rewrite the signal block, ${\bf P}$, as ${\bf P}= {\bf P}_{\rm s}+ {\bf \Psi}$, where ${\bf P}_{\rm s}$ is the signal block that is composed of the signals of interest, and $ {\bf \Psi}$ is the perturbations of signal block resulted from the interference of the high-pass filter and the AWGN noise. Then, we can rewrite the signal null-space, $\bar{\bf U}_{\rm P}$, as $\bar{\bf U}_{\rm P}=\bar{\bf U}_{{\rm P}_{\rm s}}+\Delta \bar{\bf U}_{ \Psi}$, where $\bar{\bf U}_{{\rm P}_{\rm s}}$ is the null-space that corresponds to ${\bf P}_{\rm s}$ and $\Delta\bar{\bf U}_{ \Psi}$ is the related perturbations. The perturbations of the signal null-space can be approximately written as
\begin{align}\label{DU1}
\Delta\bar{\bf U}_{ \Psi}=-{\bf U}_{\rm P}{\bf E}_{\rm P}^{-1}{\bf V}_{\rm P}^H{\bf \Psi}^H\bar{\bf U}_{{\rm P}_{\rm s}}.
\end{align}
We define a null-spectrum function as $F(f, \bar{\bf U}_{\rm P}) = {\bf p}_t(f)^H\bar{\bf U}_{\rm P}\bar{\bf U}_{\rm P}^H{\bf p}_t(f)$, which is the denominator of the objective function in \eqref{music1}. It is noted that $F(f\dop, \bar{\bf U}_{{\rm P_s}}) =0$. When the interference and noise term are introduced, the objective function in \eqref{music1} is equivalent to finding $L$ estimates of $f\dop$, such that $F(\hat f\dop, \bar{\bf U}_{\rm P})$ approach to $0$. The error between $\hat f\dop$  and $f\dop$ is denoted as $\Delta f_l$, which determines the performance of the MUSIC algorithms. For notational simplicity, we drop the subscript $l$. At a high SNR,  $\Delta f$ can be obtained via the Newton method, i.e.,
\begin{align}\label{Df}
\Delta f  = &
 {\frac{\partial F(f, \bar{\bf U}_{\rm P})}{\partial f}\over \frac{\partial^2 F(f, \bar{\bf U}_{\rm P})}{\partial^2 f}} \triangleq \frac{F_1(f, \bar{\bf U}_{\rm P})}{F_2(f, \bar{\bf U}_{\rm P})}
 \approx  \frac{F_1(f, \bar{\bf U}_{\rm P_s})+\Delta F_1}{F_2(f, \bar{\bf U}_{\rm P_s})+\Delta F_2}\notag\\
 \approx &\frac{ \Delta F_1}{F_2(f, \bar{\bf U}_{\rm P_s}) }
 =        {{{\rm Re}\left[{\bf p}_0(f)^H\Delta \bar{\bf U}_{\Psi} \bar{\bf U}_{\rm P_s}^H {\bf p}^{(1)}_0(f) \right]}\over{{{\bf p}^{(1)}_0}(f)^H \bar{\bf U}_{\rm P_s} \bar{\bf U}_{\rm P_s}^H {{\bf p}^{(1)}_0}(f)} },
\end{align}
where $F_1(\cdot)$ and $F_2(\cdot)$ denote the first- and second-order derivatives with respect to $f$,  $\Delta F_1$ denotes the error between $F_1(f, \bar{\bf U}_{\rm P })$ and $F_1(f, \bar{\bf U}_{\rm P_s})$,   $\Delta F_2$ denotes the error between $F_2(f, \bar{\bf U}_{\rm P_s})$ and $F_2(f, \bar{\bf U}_{\rm P})$, and ${{\bf p}^{(1)}_0}(f)$ denotes the first order derivative of the basis vector with respect to $f$.  The derivation can be referred to \cite{MUSICaly}. Substituting \eqref{DU1} into \eqref{Df}, we simplify the expression of $\Delta f$ as
\begin{align}\label{DeltaF}
\Delta f  = \frac{{\rm Re}\left[{\bm \beta}(f)^H{\bf\Psi}^H{\bm \gamma}(f)\right]}{{\bf p}_0^{(1)}(f)^H{\bm \gamma}(f)},
\end{align}
where ${\bm \beta}(f)$ and ${\bm \gamma}(f)$ are vectors given by $-{\bf V}_{\rm P}{\bf E}_{\rm P}^{-1}{\bf U}^H_{\rm P}{\bf p}_0(f)$ and $\bar{\bf U}_{\rm P_s} \bar{\bf U}_{\rm P_s}^H {{\bf p}^{(1)}_0}(f)$, respectively.

Intuitively, an effective way to reducing the error of $\Delta f$ is to suppress the variance of ${\bf \Psi}$. When the entries of ${\bf \Psi}$ are independent Gaussian variables with zero mean and variance of $\sigma_{\Psi}^2$, the variance of $\Delta f$ is given by $\frac12\|{\bm \beta}(f)\|_F^2\|{\bm \gamma}(f)\|_F^2\sigma^2_{\Psi}$. However, due to the CACC operation, the entries of ${\bf \Psi}$ are not independent Gaussian variables. In Appendix \ref{VarPsi}, we analyze the variance of ${\Psi}$ and obtain the following proposition.
\begin{proposition}\label{P2}
The index of the optimal reference antenna used for forming the matrix of ${\bf P}$ and $\bf Q$, $n_0$, is obtained by minimizing $\left|\sum\limits_{l=0}^L|\alpha_l|^2e^{jn_0\p_l}\right|^2$.
\end{proposition}
\begin{proof}
See proof in Appendix \ref{VarPsi}.
\end{proof}
Proposition \ref{P2} indicates that a fixed index of receiving antenna can help suppress the variance of ${\bf \Psi}$ when performing the MUSIC estimation algorithms, and hence improve the sensing performance. It is unnecessary to integrate all $n$'s and $g$'s to perform the SVD of $\bf P$. We note that $\left|\sum\limits_{l=0}^L|\alpha_l|^2e^{jn_0\p_l}\right|^2$ is equivalent to the low-pass component  in the 2D-FFT of the CACC signals. One way to obtain the optimal $n_0$ is to select the $n_0$th 2D-FFT of $\xi_n[m,g]$, such that the low-pass component of the 2D-FFT of $\xi_n[m,g]$ has the minimum among all $n$'s. This proposition is as expected since only the high-pass components contain the information of interest.

\begin{table*}[t]
	\centering
	\caption{Comparison of Computational Complexity}
	\begin{tabular}{ll}
		\bottomrule
		\multicolumn{2}{c}{The proposed Algorithm 1}\\
		\hline
		Operation &  Complexity  \\
		Conduct SVD of ${\bf P}$  & $\mathcal O\big((P+1)^2(M-P)\big)$ \\
		Obtain the objective function in \eqref{music1}    &  $\mathcal O\big(F_X(P+1)(M-P)\big)$\\
		Overall & $\mathcal O\big((P+1)^2(M-P)\big)$\\
		\bottomrule
		\multicolumn{2}{c}{Conventional MUSIC}\\
		\hline
		Operation &  Complexity  \\
		Conduct SVD &  $\mathcal O\big((P+1)^2(M-P)\big)$ \\
		Obtain the objective function   &  $\mathcal O\big(2F_X(P+1)(M-P)\big)$\\
		Overall &   $\mathcal O\big(2F_X(P+1)(M-P)\big)$\\
		\bottomrule
	\end{tabular}
	\label{Complex}
\end{table*}

\subsection{Complexity Analysis}
In this subsection, we analyze the computational complexity of Algorithm 1. Note that the delay and the Doppler frequency are estimated in parallel. Without losing generality, we only analyze the complexity of estimating Doppler frequency.  One main computation in Algorithm 1 is the SVD of ${\bf P}$  with the dimension of $(P+1)\times(M-P)$.  Since we only need the left singular matrices of ${\bf P}$, the complexities for obtaining ${\bf U}_{\rm P}$  is $\mathcal O\big((P+1)^2(M-P)\big)$. Another main computation is obtaining the objective function in \eqref{music1}. Given $F_X$ candidates of $f'$, obtaining the objective function in \eqref{music1} has a complexity of $\big(\mathcal O(F_X(P+1)(M-P))+\mathcal O(F_X(M-P))\big)\approx\mathcal O\big(F_X(P+1)(M-P)\big)$.
Other steps in Algorithm 1 have much lower complexity and can be omitted. Generally, the number of candidates, $F_X$, should be no greater than $P+1$. Hence,  the overall complexity for estimating Doppler frequency is $\mathcal O\big((P+1)^2(M-P)\big)$. Likewise,  the complexity for estimating delay is $\mathcal O\big((Q+1)^2(G-Q)\big)$. The complexities of the main steps of Algorithm 1 are summarized in Table I and are compared with those of the conventional MUSIC method.
From Table I, we note that the overall complexity of conventional MUSIC doubles our proposed mirrored-MUSIC, which is because conventional MUSIC requires a doubled number of candidates.

\section{High Resolution AoAs Estimation}\label{HAoA}

In Algorithm 1, the NLOS AoAs are still unknown but necessary for locating the targets. When the number of spatial samples is large, AoAs can be estimated directly using one-shot measurements in the spatial domain, which can be done in parallel with estimating  Doppler frequency and delay  \cite{gl,lushan5G}. When the number of antennas at the BS is small, such as the JCAS system setup in \cite{lushan}, the one-shot measurements are  insufficient for achieving accurate AoA estimation. In this section, we propose a high-resolution AoA estimation algorithm by combining measurements from the spatial and other domains.

S. Chuang et al. proposed a high-resolution AoA estimation method by using both time-domain and spatial-domain measurements \cite{high_reso}. However, their method is only applicable to narrowband systems and there exists a problem in the broadband scenario that some AoAs estimates would be missing if multiple Doppler frequencies (or delays) are close to each other. We will analyze the reason for this problem and obtain more accurate AoA estimates using measurements from all three domains in time, frequency, and space.

Using measurements in the spatial domain can distinguish among $N$ AoAs only. We attempt to equivalently enlarge the length of the spatial array response vectors by integrating both the time-domain and the frequency-domain signals into the spatial domain. Since all other parameters except NLOS AoAs have been obtained from Algorithm 1, the integrated signal vector only varies with NLOS AoAs. Otherwise, without the results from Algorithm 1, all parameters are mixed together and can be difficult to be obtained simultaneously. In the following of this section, we assume that delays and Doppler frequencies are already obtained at the BS. Still using $\xi_n[m,g]$, we generate a spatial signal vector as
\begin{align}\label{xxx}
{\bf c}[m,g] =&[\xi_1[m,g], \cdots,\xi_{N-1}[m,g] ]^T\notag\\
=&\sum_{l'=-L,l\neq0}^{L}P_{l'}{\bf a}(\p_{l'})e^{j m \bar f_{{\rm D},l'}}e^{-jg\bar\tau_{l'}},
\end{align}
where ${\bf a}(\p_{l'})=\exp[j\p_{l'}(1,\cdots, N-1)]^T$ is the $(N-1)\times1$ array response vector. Note that the length of the array response vector is reduced to $N-1$ due to the CACC operation. We integrate the spatial domain with the other two domains to enlarge the dimension of array response vectors. This can be realized using the following proposed matrix, i.e.,
\begin{align}\label{bary}
{\bf C}'[m,g]&= \left[\begin{array}{cc}
{\bf c}[m,g] & {\bf c}[m,g]\\
{\bf c}[m,g+1]&{\bf c}[m+1,g]\\
\vdots&\vdots\\
{\bf c}[m,g+C-1]&{\bf c}[m+C-1,g]\\
\end{array}\right],
\end{align}
where  $C, C\in {\mathbb N},$ satisfies $4L/(N-1)< C< \min (G-4L, M-4L)$. The dimension of ${\bf C}'[m,g]$ is $C(N-1)\times 2$. The first column of ${\bf C}'[m,g]$ is an enlarged vector corresponding to the spatial (angle) domain and the frequency (delay) domain.  The second column is an enlarged vector corresponding to the spatial (angle) domain and the time (Doppler frequency) domain.
In Appendix \ref{AppB}, we show that the basis vector for the first column of ${\bf C}'[m,g]$ is given by
\begin{align}\label{bac1}
{\bf c}^1_{l'}=\left[\begin{array}{c}
{\bf a}(\p_{l'}) e^{-j 0 \bar\tau_{l'}}\\
{\bf a}(\p_{l'}) e^{-j 1 \bar\tau_{l'}}\\
\vdots\\
{\bf a}(\p_{l'}) e^{-j (C-1) \bar\tau_{l'}}\\
\end{array}\right].
\end{align}
Likewise, the basis vector for the second column of ${\bf C}'[m,g]$, denoted as ${\bf c}^2_{l'}$,  has the same expression as \eqref{bac1} with replacing $\bar\tau_{l'}$ by $\bar f_{{\rm D},l'}$.

It is noted that the AoA parameter in \eqref{bary} is mixed with other known parameters. The other two parameters, i.e., delay and Doppler frequency, are estimated and already available. These parameters are used to enlarge the dimension of the basis vector. We stack ${\bf C}'[m,g]$ with different $m$ and $g$ into a bigger matrix,
\begin{align}\label{AssmC}
 {\bf C}=[{\bf C}'[0,0],{\bf C}'[1,1],\cdots,{\bf C}'[C_1,C_1]],
\end{align}
where $C_1, C_1\in\mathbb N,$ satisfies $C+C_1<\min(M,G).$
The basis vectors of columns of $\bf C$ are given by $ {\bf c}_{l'}^1$ and ${\bf c}_{l'}^2$, with $l'\in\{\pm1,\cdots,\pm L\}$. Hence, the rank of ${\bf C}$ is $4L$.
It is worth noting that those $4L$ basis vectors only correspond to $(L+1)$ AoAs, i.e., one LOS AoA and $L$ NLOS AoAs. The $2L$ basis vectors with ${l'}>0$ all correspond to $\Omega_0$, which is already known at the BS. For the other $2L$ basis vectors with $l'<0$, both ${\bf c}_{l'}^1$  and $ {\bf c}_{l'}^2$  correspond to one AoA.  Hence, it is not necessary to test the basis vectors with ${l'}>0$.

The AoAs can be obtained by performing the SVD of ${\bf C}$. In general, one pair of basis vectors, ${\bf c}_{l'}^1$ and ${\bf c}_{l'}^2$, corresponds to one AoA. By using the total $2L$ basis vectors, $l'\in\{1,\cdots,L\},$ we can obtain $L$ AoAs corresponding to $L$ NLOS targets. However, we notice an issue here: when the values of delays and  Doppler frequencies of different targets are close to each other, respectively, different basis vectors with similar delays and Doppler frequencies can be responded to multiple different AoAs. This would directly cause that some estimates of $\Omega_l$ with a lower gain will be missing from the estimated outputs. We call this issue as the ``ambiguity of basis vectors''. The method in \cite{high_reso} is proposed in the narrowband scenario and did not address the ambiguity of basis vectors.

Let us denote the SVD of ${\bf C}$ as ${\bf C}={\bf U}_{\rm C}{\bf E}_{\rm C}{\bf V}^H_{\rm C}$,
where ${\bf E}_{\rm C}$ is an $4L\times 4L$ diagonal matrix with the $4L$ largest entries, $ {\bf U}_{\rm C}$ is the left singular matrix of ${\bf C}$, and $ {\bf V}_{\rm C}$ is the right singular matrix of ${\bf C}$. We choose multiple peaks to address the ambiguity of basis vectors, i.e.,
\begin{align}\label{peak1}
{\rm Peak}^{X_{l'}}\left( \frac1 {\|\left[{\bf c}_{t}^1(\p',l'),{\bf c}_{t}^2(\p',l')\right]^H\bar{\bf U}_{\rm C} \|_F^2}\right), l'<0,
\end{align}
where $\bar{\bf U}_{\rm C}$ is the null-space of ${\bf U}_{\rm C}$,  $X_{l'}$ is the number of peaks in the objective function, and ${\bf c}_{t}^1(\p',l')$ and ${\bf c}^2_{t}(\p',l')$ are the testing basis vectors that have the same expressions as ${\bf c}^1_{l'}$ and ${\bf c}^2_{l'}$, respectively, after replacing $\p_{l'}$ by a testing AoA candidate,  ${\p'}$. The number of peaks is determined by how many peaks that are larger than the half of the maximal objective function. Note that the equivalent spatial dimension is extended to $C(N-1)$, hence the maximum detectable number of AoAs is $C(N-1)$. If $X_{l'}>1$, we need to retain one AoA and discard the rest. The determined AoAs with $X_{l'}=1$ form a set of $\{\p_S\}$. Supposing that the undetermined AoAs with $X_{l'}>1$ are $\hat \p_1,\cdots, \hat \p_{X_{l'}}$, we select the one with the highest value of \eqref{peak1}, dented as $\hat \p_1$, if $|\hat \p_1- \p_S|>\frac{2\pi}{C(N-1)}$. Otherwise, we select the one with a second highest value.

\begin{algorithm}[t]
	\caption{Proposed High Resolution AoA Estimation}\label{Alg_2}
	\begin{algorithmic}[1]
    \STATE {\bf Input: } $\hat{\bar \tau}_{l}$ and $\xi_n [m,g]$.
     \STATE {\bf Initialization: } $C\in {\mathbb N}$ satisfies $4L/(N-1)<C< \min(G-4L,M-4L)$. $C_1\in\mathbb N$ satisfies $C+C_1<\min(M,G).$ The AoA candidate is $\p'$ that goes through the entire range of $(-\pi,\pi)$.
     \STATE Generate ${\bf c}[m,g]$ according to \eqref{xxx}.

    \STATE Generate ${\bf C}'[m,g]$ according to \eqref{bary}.
    \STATE Generate ${\bf C}$ according to \eqref{AssmC}.

     \STATE SVD: ${\bf C}={\bf U}_{\rm C}{\bf E}_{\rm C}{\bf V}^H_{\rm C}$, where the dimension of ${\bf E}_{\rm C}$ is $4L\times 4L$.
     \STATE Denote the null-space of ${\bf U}_{\rm C}$ as $\bar{\bf U}_{\rm C}$.
     \STATE  Estimate $\{\p_l\}_{l=1}^L$  according to \eqref{peak1}.

    \STATE {\bf Output:}   $\hat {\p}_l$.
	\end{algorithmic}
\end{algorithm}

Since Algorithm \ref{Alg_2} is also based on MUSIC algorithms, the analysis in Section \ref{Accuracy} can be applied to obtaining the theoretical error of AoA estimates. Similar to \eqref{DeltaF}, letting ${\bf C}={\bf C}_s+{\bf \Psi}'$ and $\bar{\bf U}_{\rm C}=\bar{\bf U}_{\rm C_s}+{\Delta \bf U}_{\rm\Psi'}$, we can obtain the error of AoA estimates error as
\begin{align}
\Delta \p  = \frac{{\rm Re}\left[{\bm \beta}'(\p)^H{\bf\Psi}'^H{\bm \gamma}'(\p)\right]}{{\bf c}_t^{(1)}(\p,l')^H{\bm \gamma}'(\p)},
\end{align}
where ${\bf c}_t(\p,l')=[{\bf c}_{t}^1(\p,l'),{\bf c}^2_{t}(\p,l')]$, ${\bf c}_t^{(1)}(\p,l')$ is the first order derivative of ${\bf c}_t(\p,l')$ with respect to $\p$, ${\bm \beta}'(\p)$ and ${\bm \gamma}'(\p)$ are vectors that are written as $-{\bf V}_{\rm C}{\bf E}_{\rm C}^{-1}{\bf U}^H_{\rm C}{\bf c}_t(\p,l')$ and $\bar{\bf U}_{\rm C_s} \bar{\bf U}_{\rm C_s}^H {{\bf c}^{(1)}_t}(\p,l')$, respectively.

Now, we analyze the computational complexity of Algorithm 2. One main computation in Algorithm 2 is the SVD of ${\bf C}$  with the dimension of $C(N-1)\times 2C_1$.  Since we only need the left singular matrix of ${\bf C}$, the complexity for obtaining ${\bf U}_{\rm C}$  is $\mathcal O\big(2C^2(N-1)^2  C_1\big)$. Another main computation is involved in obtaining the objective function in (32). Given $C_X$ candidates of $\Omega'$, obtaining the objective function has a complexity of $\big(\mathcal O (C_XC(N-1)2C_1)+\mathcal O(2C_XC_1)\big)\approx\mathcal O\big(2C_XC(N-1)C_1\big)$.  Other steps in Algorithm 2 have much lower complexity and can be omitted. Generally, the number of candidates should be no greater than $C(N-1)$. Hence,  the overall complexity for estimating Doppler frequency is $\mathcal O\big(2C^2(N-1)^2 C_1\big)$.

\section{Simulation Results}\label{SIM}
In this section, we provide simulation results to validate the proposed scheme. The carrier frequency is $3$ GHz. The number of subcarriers is $G=256$. The frequency bandwidth is $128$ MHz. Hence, the OFDM symbol period $T$ is $2$ $u$s. The propagation delay is randomly distributed over $[0,0.4]$ $u$s. The length of CP is $D = 50$ to avoid inter-symbol interference (ISI) and the CP period $T_{\rm C}$ is about $0.4$ $u$s. The approximate interval between two packets, $T_{\rm A}$, is 1 ms. We use the preamble in $M=128$ packets for sensing parameter estimation. The velocities of targets range from -30 mps to 30 mps, and the Doppler frequency is randomly distributed over $[-0.3,0.3]$ KHz.   The AoAs of targets are random values uniformly distributed from $0$ to $ \pi$. All the targets are modeled as point sources, and the radar cross-sections are assumed to be 1. The BS employs a ULA with $N=4$ antenna elements. Unless stating otherwise, we assume that there is one LOS path and $L=3$ NLOS paths that are reflected or refracted from 3 targets. The power of the LOS path is assumed to be $10$ dB higher than those of the NLOS paths.

\begin{figure}[t]
	\centering
    \includegraphics[width=0.9\columnwidth]{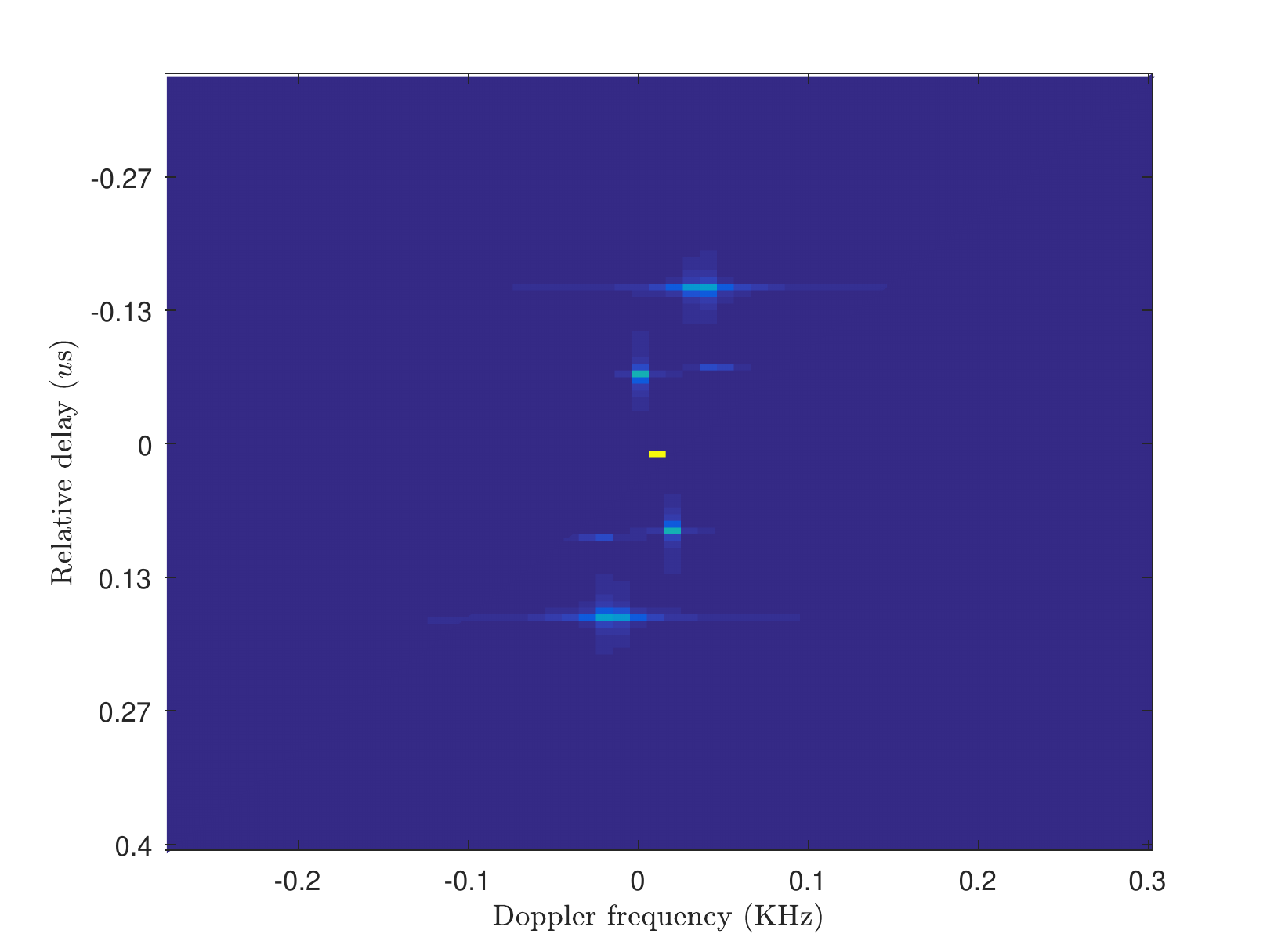}
    \caption{An example of the 2D spectrum of $\rho_n[m,g]$.}
    \label{Fig_1}
\end{figure}

\begin{figure}[ t]
    \centering
    \includegraphics[width=0.9\columnwidth]{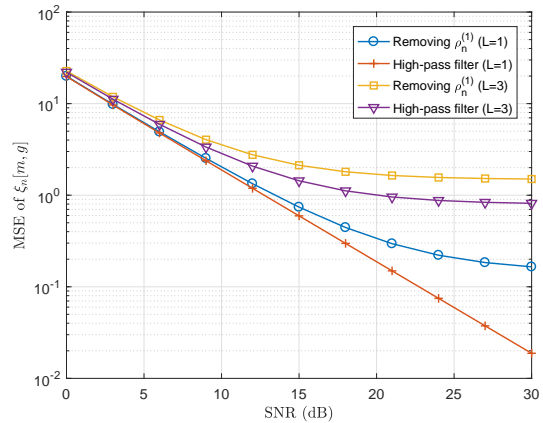}
    \caption{MSE of $\xi_n[m,g]$ versus SNR and number of paths.}
    \label{Fig_2}
\end{figure}
\begin{figure}[t]
	\centering
    \includegraphics[width=0.9\columnwidth]{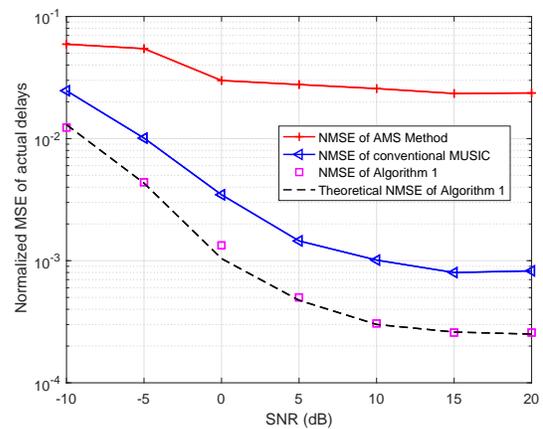}
    \caption{NMSE of the estimates for the actual delays versus SNR.}
    \label{Fig_3}
\end{figure}
\begin{figure}[t]
	\centering
    \includegraphics[width=0.9\columnwidth]{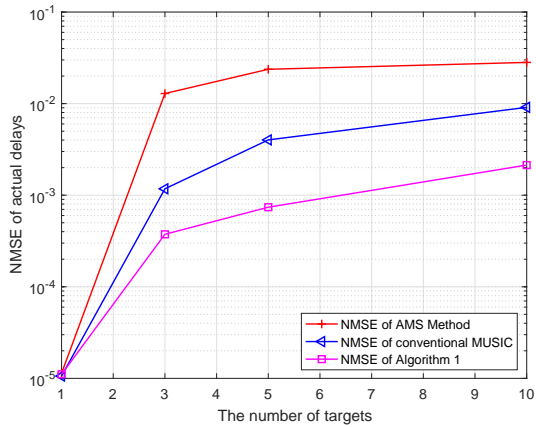}
    \caption{NMSE of the estimates for the actual delays versus $Q$. The SNR is fixed at $20$ dB.}
    \label{Fig_4}
\end{figure}

Fig. \ref{Fig_1}  illustrates an example of the spectrum of the cross-correlated signals $\rho_1[m,g]$ by performing the 2D-FFT over $m$ and $g$. The x-axis of the plot denotes the Doppler frequency, $\bar f_{{\rm D},l}$, and the y-axis denotes the relative delay, $\bar\tau_l$. In order to see the targets clearly, we truncate the absolute value of the spectrum by half. At the origin of the spectrum plot, we see that there is a rectangular spot with the highest brightness. This spot denotes the low-pass component of $\rho_n[m,g]$ which does not contain targets' information of interest. Above the origin, there are three blue spots.  We also see that there are three spots at the bottom of the figure, which denote the side products caused by cross-correlation. The side products also contain the parameters of targets with the opposite signs. This figure verifies the low-pass and the mirrored high-pass components of the CACC outputs as we described in Section III. It clearly shows that a high-pass filter can be applied to remove the non-desired low-pass component.

Fig. \ref{Fig_2} presents the MSE of $\xi_n[m,g]$, defined as $|\hat\xi_n[m,g]-\xi_n[m,g]|^2$. The MSE of $\xi_n[m,g]$ reflects the accuracy of the constructed high-pass signals and directly impacts the following sensing parameter estimation. Two methods are tested to filter out the low-pass component. One is a Butterworth filter with the cut-off frequency of $(\omega_\tau,\omega_f)=(\frac{\pi}{128},\frac{\pi}{128})$. The other one is via removing the component of the LOS path, i.e., $\rho_n^{(1)}$, over a short period from itself. It is clear that the Butterworth filter outperforms the method of removing the LOS path for either $L=1$ or $L=3$. It is worth pointing out that the MSE of the Buttworth filter drops linearly with the SNR increasing for $L=1$ target. This indicates that the input error of $\xi_n[m,g]$ can be sufficiently small when there is only one target. This is because, when $L=1$, $\rho_n^{(2)}[m,g]$ is also a low-pass component, and when there are multiple targets, the MSE approaches to a fixed level that is the mean power of $\rho_n^{(2)}[m,g]$.

Next, we present the estimation performance for sensing parameters. Since the Doppler frequencies are obtained by the same algorithm as delays, we mainly present the simulation results for delays and AoAs next.

Fig. \ref{Fig_3} illustrates the normalized MSE (NMSE) for the estimates of the actual propagation delays, i.e., $|\hat \tau_l-\tau_l|^2/T^2$.  The benchmark algorithms for comparison include the AMS method in \cite{widar2.0} and conventional MUSIC. The system setup is the same as that in Fig. \ref{Fig_1}. For our proposed Algorithm 1, the parameter $Q$, which defines the length of the mirrored vector in \eqref{qq}, is set to $128$. We see that our proposed algorithm outperforms the other two methods significantly. The AMS algorithm shows a large NMSE, possibly due to its inefficiency in dealing with multiple paths. Our proposed algorithm achieves lower NMSE than the conventional MUSIC because it removes the mirrored side products and reduces the rank of the signal space. An error floor can be observed for both our proposed Algorithm \ref{Alg_1} and the conventional MUSIC. This is caused by the error of $\xi_n[m,g]$, which cannot be removed by increasing SNR. We also plot the theoretical NMSE of the proposed mirrored-MUSIC. The NMSE of the proposed Algorithm \ref{Alg_1} matches the theoretical NMSE tightly. Our proposed algorithm can achieve better performance by using larger bandwidth since the bandwidth has a significant impact on sensing performance and mainly influences the time resolution.  With the used bandwidth increasing, the time duration of each symbol is reduced and the time resolution is improved.

Fig. \ref{Fig_4} shows how the NMSE of delay estimates varies with $L$. The system setups are the same as those in Fig. 3, except that the SNR is fixed at $20$ dB. The number of targets, $L$, is chosen from $1$ to $10$. We compare our proposed mirrored-MUSIC with the conventional MUSIC and the AMS method. It is noted that, when $L$ is 1, the NMSEs for all these methods are nearly the same. When $L$ ranges from $3$ to $10$,  our proposed mirrored-MUSIC can achieve much lower NMSE than the other two methods. The NMSE increases with the number of targets while the growth rate drops, because the delays of multiple targets become closer to each other and  can be separated into several groups.

\begin{figure}[t]
	\centering
    \includegraphics[width=0.9\columnwidth]{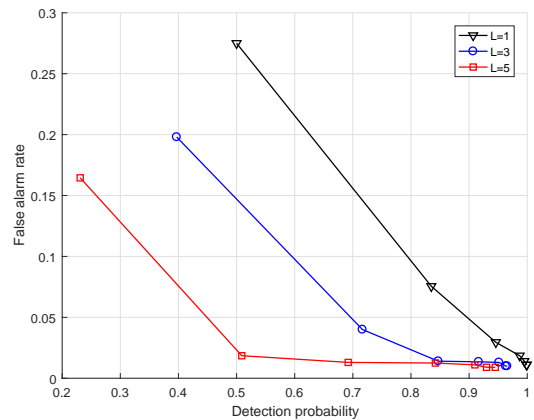}
    \caption{ROC versus SNR and the number of targets. Each mark on one curve is for one SNR value from -10 dB (left) to 20 dB (right), at a step of 5 dB.}
    \label{Fig_5}
\end{figure}
\begin{figure}[t]
    \centering
    \includegraphics[width=0.9\columnwidth]{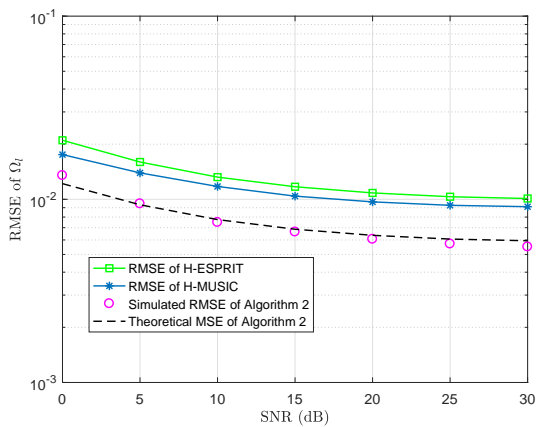}
    \caption{RMSE of the AoA estimates versus SNR.}
    \label{Fig_6}
\end{figure}

Fig. \ref{Fig_5} presents the receiver operating characteristic (ROC) of the proposed Algorithm 1 after matching the estimates for delays and Doppler frequencies. The x-axis measures the detection probability and the y-axis measures the false-alarm rate. We use a common threshold of NMSE, $10^{-3}$, to detect whether the estimated values are effective ones. If the NMSE is smaller than $10^{-3}$, the estimated values are seen as effective values, and the targets are regarded as detected. We increase the SNR from $-10$ dB to $20$ dB, at an interval of 5 dB, to observe the variance of the ROC. According to the figure, the detection probability rises with the SNR increasing and approaches to one, while the false-alarm rate shows an opposite trend and approaches to zero. With the number of targets, $L$, increasing, the detection probability drops and the false-alarm rate grows, which indicates that the performance declines. We can also see that the sum of detection probability and the false-alarm rate is less than 1, which is because one target may have more than one estimate.

Fig. \ref{Fig_6} shows the RMSE of the AoA estimation versus SNR. Our proposed Algorithm \ref{Alg_2} is compared with the H-MUSIC method and the H-ESPRIT method from \cite{high_reso}. The system setup is the same as that in Fig. \ref{Fig_1}. Noting that the H-MUSIC and the H-ESPRIT are based on hybrid arrays, we simplified their methods by letting each hybrid array has one receiving antenna, which is equivalent to a fully-digital array. For the initialization of Algorithm \ref{Alg_2}, we let $C=128$ and $C_1=10$, and use the estimated delays obtained from Algorithm \ref{Alg_1}.   Our algorithm outperforms H-MUSIC and H-ESPRIT, since our proposed Algorithm \ref{Alg_2} can better address the ambiguity of basis vectors as we discussed in the section V.
We see that the RMSE of AoAs drops slighly with SNR increasing from $15$ dB to $30$ dB, and the simulated RMSE matches well with the theoretical values. To further decrease the RMSE, we can increase the bandwidth to increase the accuracy of all estimates.

\section{Conclusion}\label{sec-conc}
We have proposed an uplink sensing scheme for JCAS PWNs, which can achieve high-accuracy sensing parameter estimation with asynchronous transceivers and a small number of receiving antennas. We extend the CACC methods to mitigate the timing and frequency ambiguity. We then propose a mirrored-MUSIC algorithm to efficiently handle the CACC outputs with equivalently doubled unknown sensing parameters, at a complexity lower than that of the conventional MUSIC algorithm. Simulation results demonstrate that our proposed mirrored-MUSIC can effectively estimate the actual values of delay and Doppler frequency. Using the estimates of delay and Doppler frequency, we then obtain high-resolution AoAs estimates with a small number of receiving antennas. This is achieved via an improved MUSIC algorithm that combines measurements from spatial, temporal, and frequency domains. Our scheme enables radar sensing to be effectively implemented in mobile networks using the uplink channel and requires little modifications on infrastructure or advanced hardware, such as a full-duplex transceiver. The results show that the proposed scheme of uplink parameter estimation outperforms the state of the arts and can accurately detect multiple targets for the JCAS PWNs.

Our proposed mirrored-MUSIC algorithm can be applied to other problems with mirrored signals, such as the general harmonic retrieval problem with sinusoidal modulations. Its basic idea can also be extended to other spectral analysis techniques such as ESPRIT and the matrix pencil method. The proposed high-resolution AoA estimation can also be applied to other problems involving a similar combination of multi-domain measurements.

\begin{appendices}
\section{Proof of Theorem \ref{T1}}\label{T1proof}
We write $\rho_n^{(1)}[m,g]$ as
\begin{align}
\rho_n^{(1)} &=D_n[m,g]D_0^H[m,g]=  |\alpha_0|^2 e^{jn\p_0}.
\end{align}
It is noted that $\rho_n^{(1)}$ is invariant with $m$ and $g$.

We express the sum of $\rho^{(3)}_n[m,g]$ and $\rho^{(4)}_n[m,g]$ as
\begin{align}
&\rho_n^{(3)}[m,g]+\rho_n^{(4)}[m,g]\notag\\
           =&\sum\limits_{l=1}^L\alpha_0\alpha_l^He^{j2\pi mT_{\rm A}(-f\dop)}e^{-j\frac{2\pi g}{T}(\tau_0-\tau_l)}e^{jn\p_0} +\notag\\
         &  \sum\limits_{l=1}^{L}\alpha_l\alpha_0^He^{j2\pi mT_{\rm A}f\dop}e^{-j\frac{2\pi g}{T}(\tau_l-\tau_0)}e^{jn\p_l}.
\end{align}
The 2D-FFT of $\rho_n^{(3)}[m,g]+\rho_n^{(4)}[m,g]$ is given by
\begin{align}
&{\rm FFT}\left(\rho_n^{(3)}[m,g]+\rho_n^{(4)}[m,g]\right)\notag\\
           =&\alpha_0\alpha_l^He^{jn\p_0}\frac{\sin\left(MT_{\rm A}f\dop+m'\right)\pi}{\sin\left(T_{\rm A}f\dop+\frac{m'}M\right)\pi}\frac{\sin\left(G\frac{\tau_0-\tau_l}{T}+g'\right)\pi}{\sin\left(\frac{\tau_0-\tau_l}{T}+\frac{g'}G\right)\pi}+\notag\\
           &\alpha_l\alpha_0^He^{jn\p_l}\frac{\sin\left(MT_{\rm A}f\dop-m'\right)\pi}{\sin\left(T_{\rm A}f\dop-\frac{m'}M\right)\pi}\frac{\sin\left(G\frac{\tau_0-\tau_l}{T}-g'\right)\pi}{\sin\left(\frac{\tau_0-\tau_l}{T}-\frac{g'}G\right)\pi},
\end{align}
and has an impulsive shape with the peak points at
\begin{align}
(\omega_f,\omega_\tau)=\pm (\pi T_{\rm A} f\dop ,\pi(\tau_0-\tau_l)/T).
\end{align}
Note that $m'$ and $g'$ are integers, while $\omega_f$ and $\omega_\tau$ are continuous values ranging from $-\pi$ to $\pi$.

As for $\rho_n^{(2)}[m,g]$, it is written as
\begin{align}
 &\rho_n^{(2)}[m,g]\notag\\
 =&I_n[m,g]I_0^H[m,g]\notag\\
                  =&\big(\sum\limits_{l=1}^{L} \alpha_l e^{jn\p_l}e^{j2\pi mT_{\rm A}(f\dop+\delta_f(m))}e^{-j \frac{2\pi g}{T}(\tau_l+\delta_\tau(m))}\big)\times\notag\\
                  &\big(\sum\limits_{l=1}^{L} \alpha^H_l e^{-j0\p_l}e^{-j2\pi mT_{\rm A}(f\dop+\delta_f(m))}e^{j\frac{ 2\pi g}{T}(\tau_l+\delta_\tau(m))}\big)\notag\\
                  =&\sum\limits_{l=1}^L|\alpha_l|^2e^{jn\p_l}+ \sum\limits_{l=1}^L\sum\limits_{x\neq l}^L \alpha_l\alpha_x^He^{jn\p_l}e^{j2\pi mT_{\rm A}f_{l,x}}e^{-j\frac{2\pi g}{T}\tau_{l,x}}\notag\\
                  \triangleq& \bar\rho^{(2)}+\tilde \rho^{(2)}[m,g],
\end{align}
where $f_{l,x}=f_{{\rm D},l}-f_{{\rm D},x}$ and $\tau_{l,x}=\tau_l-\tau_x$. It is noted that $\bar\rho^{(2)}$ is an invariant component and $\tilde \rho^{(2)}[m,g]$ is a variant component. The 2D-FFT of $\tilde\rho^{(2)}[m,g]$ also has an impulsive shape, but the power of $\tilde\rho^{(2)}[m,g]$ is  significantly lower than other components  and can be  neglected.

\section{The Error of $\xi_n[m,g]$}\label{MSEofXI}
 The  components of  $\rho_n^{(1)}$ and  $\bar\rho_n^{(2)}$ are invariant with $m$ and $g$ and can be largely suppressed after high-pass filtering. Without the noise term, the error of $\xi_n[m,g]$ is mainly caused by the variant part of $\tilde\rho_n^{(2)}[m,g]$. Hence, we can define the statistical mean error of $\xi_n[m,g]$ as
\begin{align}
{\delta}_{\xi}={\rm var}(\tilde\rho_n^{(2)}[m,g]).
\end{align}

As for the AMS method, it assumes that $|\alpha_0|\gg|\alpha_l|$ and estimates $D_n[m,g]$ as the mean value of $y_n[m,g]$, which is denoted as $\hat D_n[m,g]$. Then, the AMS method obtains two signals. One is $A_n[m,g]=y_n[m,g]-\hat D_n[m,g]=\hat I_n[m,g]$ and the other one is $B_n[m,g]=y_n[m,g]+\hat D_n[m,g]\approx2\hat D_n[m,g]+\hat I_n[m,g]$. The cross correlated signal  of the AMS method is expressed as
\begin{align}
\xi_n^{\rm AMS}[m,g]&=A_n[m,g]B_0^H[m,g]\notag\\
&\approx  I_n[m,g](2D_0[m,g]+I_0[m,g])\notag\\
&\approx2\rho_n^{(4)}[m,g]+\rho^{(2)}_n[m,g].
\end{align}
The AMS method uses $2\rho_n^{(4)}[m,g]$ as output to conduct parameter estimation. The error in the AMS method is up to $\rho^{(2)}_n[m,g]$, which is larger than the $\xi_n[m,g]$ used in our method without by-product suppression. More importantly,   $D_n[m,g]$ and $I_n[m,g]$ in the AMS method are  approximated values, which would result in extra errors. Therefore, the AMS method causes larger $\xi_n[m,g]$ errors than our method in \eqref{cros}.

\section{The Variance of $\bf\Psi$}\label{VarPsi}
The dimension of ${\bf \Psi}$ is $(P+1)\times(M-P)$. The expectation of entries of $\bf\Psi$ can be approximately as
\begin{align}\label{EPsi}
&{\mathbb E}[{\bf\Psi}]_{p_i,m_j}\notag\\
=&{\mathbb E}\big[\rho_{n_0}[m_j+p_i,g_0]+\rho_{n_0}[m_j+P-p_i,g_0]-2\bar\rho_{n_0}\notag\\
&-\rho^{(3)}_{n_0}[m_j+p_i,g_0]-\rho^{(4)}_{n_0}[m_j+p_i,g_0]\notag\\
&-\rho^{(3)}_{n_0}[m_j+P-p_i,g_0]-\rho^{(4)}_{n_0}[m_j+P-p_i,g_0]\big]\notag\\
\approx &2{\mathbb E}\left[\tilde\rho_{n_0}^{(2)}[m,g_0]\right]+2{\mathbb E}\left[D_{n_0}[m,g_0]+I_{n_0}[m,g_0]\right]{\mathbb E}[z_{n_0}[m,g_0]]\notag\\
=&2{\mathbb E}\left[\tilde\rho_{n_0}^{(2)}[m,g_0]\right]+2{\mathbb E}\left[\rho_n^{(1)}+\bar\rho_n^{(2)} \right]{\mathbb E}[z_{n_0}[m,g_0]]\notag\\
=&2{\mathbb E}\left[\tilde\rho_{n_0}^{(2)}[m,g_0]\right]+2 \sum\limits_{l=0}^L|\alpha_l|^2e^{jn_0\p_l}{\mathbb E}[z_{n_0}[m,g_0]],
\end{align}
where the definition of $\tilde\rho_{n}^{(2)}[m,g]$ and $\bar\rho_n^{(2)}$ can be referred to Appendix \ref{T1proof}. The first term in the end of  \eqref{EPsi}, $2{\mathbb E}\left[\tilde\rho_{n_0}^{(2)}[m,g_0]\right]$,  is the interference with variance of $\delta_\xi$ after conducting the high-pass filter,   and the second term in \eqref{EPsi} denotes the noise term after CACC. Hence, the variance of each entry of ${\bf\Psi}$ is
\begin{align}
&{\rm var}\big([{\bf\Psi}]_{p_i,m_j}\big)\notag\\
=&{\rm var}\left[2\tilde\rho_{n_0}^{(2)}[m,g_0]\right]+  {\rm var}\left[2\sum\limits_{l=0}^L|\alpha_l|^2e^{jn_0\p_l}z_{n_0}[m,g_0]\right]\notag\\
=&4\delta_\xi+4\left|\sum\limits_{l=0}^L|\alpha_l|^2e^{jn_0\p_l}\right|^2\sigma^2\notag\\
\triangleq & 4\delta_\xi + 4 \delta_{n_0} \sigma^2.
\end{align}
To minimize the variance of  each entry of ${\bf\Psi}$, $\delta_{n_0}$ needs to be minimized. Therefore, the optimal selected index of receiving antenna, $n_0$, satisfies that $\left|\sum\limits_{l=0}^L|\alpha_l|^2e^{jn_0\p_l}\right|^2$ is minimized.
\section{Proof of \eqref{bac1}}\label{AppB}
From \eqref{xxx}, the basis vectors of ${\bf c}[m,g]$ are given by ${\bf a}(\p_{l'})$. It is noted that
\begin{align}
{\bf c}[m,g+c]=\sum\limits_{l'=-L,l\neq0}^{L}\psi_{l'}[m,g]{\bf a}(\p_{l'})e^{-j c\bar\tau_{l'}},
\end{align}
where $\psi_{l'}[m,g]=P_{l'}e^{j m \bar f_{{\rm D},l'}}e^{-j g\bar\tau_{l'}}$ denotes the weighting factor. Hence, the basis vectors of ${\bf c}[m,g+c]$ only have a phase shift of $-c\bar\tau_{l'}$ compared with those of ${\bf c}[m,g]$. We denote the first column of ${\bf C'}[m,g]$ as ${\bf c}_1[m,g]$. Then, we have
\begin{align}
 &{\bf c}_1[m,g]\notag\\
 =&\big[{\bf c'}^T[m,g],\cdots,{\bf c'}^T[m,g+C-1]\big]^T\notag\\
=&\sum\limits_{l=-L}^{L}\psi_{l}[m,g]\big[{\bf a}^T(\p_l)e^{-j 0\bar\tau_l},\cdots,{\bf a}^T(\p_l)e^{-j (C-1)\bar\tau_l}\big]^T.
\end{align}
Therefore, the basis vectors for ${\bf c}_1[m,g]$  are given by $[{\bf a}^T(\p_{l'}),\cdots,{\bf a}^T(\p_{l'})e^{-j (C-1)\bar\tau_{l'}}]^T$. Likewise, the basis vectors for the second column of ${\bf C}'[m,g]$ have the same expression with replacing $\bar\tau$ by $\bar f_{\rm D}$.
\end{appendices}

\bibliographystyle{IEEEtran}

\end{document}